\documentclass[runningheads,english]{llncs}

\usepackage[utf8]{inputenc}
\usepackage{algorithm}
\usepackage[noend]{algpseudocode}
\usepackage{amsmath}
\usepackage{amssymb}
\usepackage{booktabs}
\usepackage{cite}
\usepackage{caption}
\usepackage{csquotes}
\usepackage{graphicx}
\usepackage{hyperref}
\usepackage[autolanguage]{numprint}
\usepackage{subcaption}
\usepackage{tikz}
\usetikzlibrary{patterns}
\usepackage[disable]{todonotes}
\usepackage{url}
\usepackage{xcolor}

\usepackage{xspace}

\newcommand{\ie}{i.\,e.,\xspace}
\newcommand{\eg}{e.\,g.,\xspace}

\newcommand{\aka}{a.\,k.\,a.\xspace}

\newcommand{\etal}{et al.\xspace}

\newcommand{\Oh}{\ensuremath{\mathcal{O}}}

\newcommand{\adapt}{ADS\xspace}
\newcommand{\bc}{BC\xspace}
\newcommand{\mathbc}{\mathbf{b}}
\newcommand{\approxbc}{\widetilde{\mathbf{b}}}

\newcommand{\tool}[1]{\textsf{#1}}
\newcommand{\nwk}{\tool{NetworKit}\xspace}
\newcommand{\kad}{\tool{KADABRA}\xspace}
\newcommand{\rk}{\tool{RK}\xspace}

\newcommand{\darkgreen}{black!60!green}

\newcommand{\overallSuVsOriginalSingle}{6.89846099713451}

\newcommand{\adsSuLfVsOriginal}{57.34439224050918}
\newcommand{\adsSuSfVsOriginal}{65.29488898389054}
\newcommand{\overallSuOmpVsOriginal}{13.544344311728693}
\newcommand{\adsSuOmpVsOriginal}{22.66442436431893}

\newcommand{\adsParallelSuNaive}{6.27616521473421}
\newcommand{\adsParallelSuLf}{15.879639123178466}
\newcommand{\adsParallelSuSf}{18.081267115073295}
\newcommand{\adsParallelSuDt}{10.818961592817336}
\newcommand{\adsSuLfVsOmp}{2.5301499530156892}
\newcommand{\adsSuSfVsOmp}{2.8809418644086824}
\newcommand{\adsSuDtVsOmp}{1.7238172072682618}
\newcommand{\numInst}{27}

\DeclareMathOperator{\rlxmove}{\gets_\mathsf{relaxed}}
\DeclareMathOperator{\loadacq}{\gets_\mathsf{acquire}}
\DeclareMathOperator{\storerel}{\gets_\mathsf{release}}

\begin{document}
\mainmatter
\title{Parallel Adaptive Sampling \\ with almost no Synchronization\thanks{Partially supported by 
grant ME 3619/3-2 within German Research Foundation (DFG) Priority 
Programme 1736 \emph{Algorithms for Big Data}. 
}}
\titlerunning{Parallel Adaptive Sampling}

\author{Alexander van der Grinten\and Eugenio Angriman\and Henning Meyerhenke}
\authorrunning{van der Grinten \etal}
\institute{Department of Computer Science, Humboldt-Universität zu Berlin, Germany\\
\email{\{avdgrinten,angrimae,meyerhenke\}@hu-berlin.de}}

\maketitle

\begin{abstract}
Approximation via sampling is a widespread technique whenever exact solutions
are too expensive. In this paper, we present techniques for an efficient parallelization of adaptive (\aka progressive) sampling algorithms
on multi-threaded shared-memory machines.
Our basic algorithmic technique requires no synchronization
except for atomic \texttt{load-acquire} and \texttt{store-release}
operations.
It does, however, require $\Oh(n)$ memory per thread,
where $n$ is the size of the sampling state.
We present variants of the algorithm that either reduce this memory consumption
to $\Oh(1)$ or ensure that deterministic results are obtained.

Using the \kad algorithm for betweenness centrality (a popular measure
in network analysis) approximation
as a case study, we demonstrate the empirical performance
of our techniques.
\nprounddigits{1}
In particular, on a 32-core machine, our best algorithm is $\numprint{\adsSuSfVsOmp} \times$ faster
than what we could achieve using a straightforward OpenMP-based parallelization
and $\numprint{\adsSuSfVsOriginal} \times$
faster than the existing implementation of \kad.
\npnoround
\end{abstract}

\section{Introduction}
\label{sec:intro}
%
When a computational problem cannot be solved exactly within the desired time
budget, a frequent solution is to employ approximation algorithms~\cite{Gonzalez:2007:HAA:1199638}.
With large data sets being the rule and not the exception today,
approximation is frequently applied, even to polynomial-time problems~\cite{borassi2016kadabra}.
We focus on a particular subclass of approximation algorithms:
\emph{sampling algorithms}. They sample data according to some (usually algorithm-specific)
probability distribution, perform some computation on the sample and induce a result for the full data set.

More specifically, we consider \emph{adaptive} sampling (\adapt) algorithms
(also called \emph{progressive} sampling algorithms).
Here, the number of samples that are required is not statically
computed (\eg from the input instance) but also depends on
the data that has been sampled so far. While non-adaptive sampling algorithms 
can often be parallelized trivially by drawing multiple samples in parallel,
adaptive sampling constitutes a challenge for parallelization: 
checking the stopping condition of an \adapt algorithm
requires access to all the data generated so far and thus
mandates some form of synchronization.

\paragraph*{Motivation and Contribution.}
%
Our initial motivation was a parallel implementation of the sequential 
state-of-the-art approximation algorithm \kad~\cite{borassi2016kadabra} for betweenness centrality
(\bc) approximation. \bc is a very popular centrality measure in network analysis, see Section~\ref{sub:bc-approx} for more details.
To the best of our knowledge, parallel adaptive sampling has not 
received a generic treatment yet. Hence, we propose techniques 
to parallelize \adapt algorithms in a generic way, while scaling
to large numbers of threads.
%
%
%
While we turn to \kad to demonstrate the effec\-tive\-ness of the proposed algorithms, 
our techniques can be adjusted easily to other \adapt algorithms. 

%
We introduce two new parallel \adapt algorithms,
which we call \emph{local-frame} and
\emph{shared-frame}.
Both algorithms try to avoid extensive synchronization when checking the
stopping condition. This is done by maintaining multiple copies of the sampling state
and ensuring that the stopping condition is never checked on a copy
of the state that is currently being written to.
%
%
%
\emph{Local-frame} is designed to use the least amount of synchronization possible
-- at the cost of an additional memory footprint of
$\Theta(n)$ per thread, where $n$ denotes the size of the sampling state.
This algorithm performs only atomic \texttt{load-acquire} and
\texttt{store-release} operations for synchronization, but no
expensive read-modify-write operations (like \texttt{CAS} or \texttt{fetch-add}).
%
%
%
\emph{Shared-frame}, in turn, aims instead at meeting a
desired tradeoff between memory footprint and synchronization overhead.
In contrast to \emph{local-frame}, it requires only $\Theta(1)$
additional memory per thread, but
uses atomic read-modify-write operations (\eg \texttt{fetch-add}) to accumulate samples.
We also propose the deterministic \emph{indexed-frame} algorithm; it
guarantees that the results of two different executions is the same
for a fixed random seed, regardless of the number of threads.

%
\nprounddigits{1}
Our experimental results show that local-frame, shared-frame and indexed-frame
achieve parallel speedups of $\numprint{\adsParallelSuLf}\times$,
$\numprint{\adsParallelSuSf}\times$,
and $\numprint{\adsParallelSuDt}\times$ on 32 cores, respectively.
Using the same number of cores, our OpenMP-based parallelization (functioning as a baseline)
only yields a speedup of $\numprint{\adsParallelSuNaive}\times$;
thus our algorithms are up to $\numprint{\adsSuSfVsOmp}\times$ faster.
Moreover, also due to implementation improvements and parameter tuning,
our best algorithm performs adaptive sampling $\numprint{\adsSuSfVsOriginal}\times$
faster than the existing implementation of \kad (when all implementations use 32 cores).

\npnoround

\section{Preliminaries and Baseline for Parallelization}
\label{sec:prelim}
%
\subsection{Basic Definitions}
\label{sub:basic-defs}
\paragraph{Memory model.} Throughout this paper, we target a multi-threaded shared-memory machine
with $T$ threads.
We work in the C11 memory model~\cite{ISO:2012:III} (more details in Appendix~\ref{app:c11_model});
in particular, we assume the existence of the usual atomic operations, as well as
\texttt{load-acquire} and \texttt{store-release} barriers.

\paragraph{Adaptive sampling.} 
\begin{algorithm}[bt]
\caption{Generic Adaptive Sampling}
\label{algo:adaptive}
\begin{minipage}[t]{0.50\textwidth}
Variable initialization:
\begin{algorithmic}
	\State $d \gets$ new sampling state structure
	\State $d\mathtt{.data} \gets (0, \ldots, 0)$ \Comment{Sampled data.}
	\State $d\mathtt{.num} \gets 0$ \Comment{Number of samples.}
\end{algorithmic}
\end{minipage}
\hspace{0.05\textwidth}
\begin{minipage}[t]{0.45\textwidth}
Main loop:
\begin{algorithmic}
	\While{\textbf{not} \Call{checkForStop}{$d$}}
		\State $d\mathtt{.data} \gets d\mathtt{.data} \circ \Call{sample}{\null}$
		\State $d\mathtt{.num} \gets d\mathtt{.num} + 1$
	\EndWhile
\end{algorithmic}
\end{minipage}
\end{algorithm}
For our techniques to be applicable, 
we expect that an \adapt
algorithm behaves as depicted
in Algorithm~\ref{algo:adaptive}: it
iteratively samples data (in \Call{sample}{}) and aggregates it
(using some operator $\circ$), until a stopping condition (\Call{checkForStop}{})
determines that the data sampled so far is sufficient to
return an approximate solution within the required accuracy.
This condition does not only consider the number of samples
($d\mathtt{.num}$),
but also the sampled data
($d\mathtt{.data}$). Throughout this paper, we denote the size
of that data (\ie the number of elements of $d.\mathtt{data}$) by $n$.
We assume that the stopping condition needs to be checked on
a \emph{consistent} state, \ie a state of $d$ that can
occur in a sequential
execution.\footnote{That is, $d\mathtt{.num}$ and all entries of $d\mathtt{.data}$ must
result from an integral sequence of samples; otherwise, parallelization
would be trivial.}
Furthermore, to make parallelization feasible at all, we need to assume
that $\circ$ is associative.

\subsection{Betweenness Centrality and its Approximation}
\label{sub:bc-approx}
\emph{Betweenness Centrality} (\bc) is one of the most popular vertex centrality measures
in the field of network analysis. Such measures indicate the importance of a vertex
based on its position in the network~\cite{boldi2014axioms}
(we use the terms \emph{graph} and \emph{network} interchangeably).
Being a centrality measure, \bc constitutes a function
$\mathbc: V \to \mathbb{R}$ that maps each vertex of a graph $G = (V, E)$ to
a real number -- higher numbers represent higher importance.
To be precise, the \bc of $u \in V$ is defined as
$
  \mathbc(u) = \sum_{s \neq t \in V \setminus \{u\}} \frac{\sigma_{st}(u)}{\sigma_{st}}, 
$
where $\sigma_{st}$ is the number of shortest $s$-$t$-paths
and $\sigma_{st}(u)$ is the number of shortest $s$-$t$-paths
that contain $u$.

Unfortunately, \bc is rather expensive to compute: the standard exact
algorithm~\cite{brandes2001faster} has time complexity $\Theta(|V||E|)$ for unweighted graphs.
Moreover, unless the Strong Exponential Time Hypothesis 
fails, this asymptotic running time cannot be improved~\cite{borassi2016into}.
Numerous approximation algorithms for \bc have thus been developed
(we refer to Section~\ref{sec:related_work} for an overview).
The state of the art of these approximation algorithms is the $\kad$
algorithm~\cite{borassi2016kadabra} of Borassi and Natale,
which happens to be an \adapt algorithm.
With probability $(1 - \delta)$,
\kad approximates the \bc values of the vertices within an additive error of $\epsilon$
in nearly-linear time complexity, where $\epsilon$ and $\delta$ are user-specified constants.

While our techniques apply to any \adapt algorithm,
we recall that, as a case study, we focus on scaling the \kad algorithm to
a large number of threads.

\subsection{The \kad algorithm}
\label{sub:kadabra_algo}
\kad samples vertex pairs
$(s, t)$ of $G = (V, E)$ uniformly at random
and then selects  a shortest $s$-$t$-path uniformly at random
(in \Call{sample}{} in Algorithm~\ref{algo:adaptive}).
After $\tau$ iterations, this results in a sequence of randomly selected shortest paths
$\pi_1, \pi_2, \dots, \pi_\tau$;
from those paths, \bc is estimated as:

\begin{equation*}
  \approxbc(v) = \frac{1}{\tau}\sum_{i = 1}^{\tau}x_i(v),
  \quad
  x_i(v) = 
  \begin{cases}
    1 & \text{ if } v \in \pi_i\\
    0 & \text{otherwise.}
  \end{cases}
\end{equation*}
$\sum_{i = 1}^\tau x_i$ is exactly the sampled data ($d\mathtt{.data}$)
that the algorithm has to store (\ie the accumulation $\circ$ in
Algorithm~\ref{algo:adaptive} sums $x_i$ over $i$).
To compute the stopping condition (\Call{checkForStop}{} in Algorithm~\ref{algo:adaptive}),
\kad maintains the invariants
\begin{equation}
\label{eq:kad-inv}
  \Pr(\mathbc(v) \le \approxbc(v) - f) \le \delta_{L}(v)
  \text{ ~and~ }
  \Pr(\mathbc(v) \ge \approxbc(v) + g) \le \delta_{U}(v)
\end{equation}
for two functions
$f = f(\approxbc(v), \delta_{L}(v), \omega, \tau)$ and
$g = g(\approxbc(v), \delta_{U}(v), \omega, \tau)$
depending on a maximal number $\omega$ of samples and
per-vertex probability constants $\delta_L$ and $\delta_U$
(more details in the original paper~\cite{borassi2016kadabra}).
The values of those constants are computed in a preprocessing phase
(mostly consisting of computing an upper bound on the diameter of the graph).
$\delta_L$ and $\delta_U$ satisfy
$\sum_{v \in V} \delta_L(v) + \delta_U(v) \leq \delta$
for a user-specified parameter $\delta \in (0, 1)$.
Thus, the algorithm terminates once $f, g < \epsilon$; 
the result is correct with an absolute error of $\pm \epsilon$ and probability $(1 - \delta)$.
We note that checking the stopping condition of \kad
on an inconsistent state
leads to incorrect results.
For example, this can be seen from the fact that
$g$ is increasing with $\approxbc$ and decreasing with $\tau$, see Appendix~\ref{app:kadabra}.

\subsection{First Attempts at \kad Parallelization}
\label{sub:first-attempts}
In the original \kad implementation\footnote{
Available at: \url{https://github.com/natema/kadabra}},
a lock is used to synchronize
concurrent access to the sampling state.
As a first attempt to improve the scalability,
we consider an algorithm that iteratively computes a fixed number
of samples in parallel (\eg using an OpenMP \texttt{parallel for} loop),
then issues a synchronization barrier (as implied by the \texttt{parallel for} loop)
and checks the stopping condition afterwards. While sampling,
atomic increments are used to update the global sampling data.
This algorithm is arguably the \enquote{natural} OpenMP-based parallelization
of an \adapt algorithm and can be implemented in
a few extra lines of code. Moreover, it already improves upon
the original parallelization. However, as
shown by the experiments in Section~\ref{sec:experiments}, further significant
improvements in performance are possible by switching
to more lightweight synchronization.

\section{Scalable Parallelization Techniques}
\label{sec:algorithms}
To improve upon the OpenMP parallelization from Section~\ref{sub:first-attempts},
we have to avoid the synchronization barrier before the stopping condition can be
checked.
This is the objective of our \emph{epoch-based}
algorithms that constitute the main contribution of this paper.
In Section~\ref{sub:epoch_based}, we formulate the main idea of our algorithms
as a general framework and prove its correctness. The subsequent subsections present
specific algorithms based on this framework and discuss tradeoffs between them.

\subsection{Epoch-based Framework}
\label{sub:epoch_based}

\begin{figure}[t]%
\begin{subfigure}[t]{.4\textwidth}
\quad
\fbox{\begin{tabular}{r@{\quad}l@{\quad}l}
	\textsf{int} & \texttt{epoch} &$\gets e$ \\
	\textsf{int} & \texttt{num} &$\gets 0$\\
	\textsf{int} & \texttt{data}[$n$] &$\gets (0, \ldots, 0)$
\end{tabular}}%
\caption{Structure of a state frame (SF) for epoch $e$.
	\texttt{num}: Number of samples, \texttt{data}: Sampled data}%
\label{fig:sf}
\end{subfigure}%
\begin{subfigure}[t]{.6\textwidth}
\qquad
\qquad
\begin{tabular}{r@{\quad}l@{\quad}l}
	\textsf{bool} & \texttt{stop} &$\gets$ \textsf{false}\\
	\textsf{int} & \texttt{epochToRead} &$\gets 0$\\
	SF $\ast$ & \texttt{sfFin[$T$]} &$\gets (\textsf{null}, \ldots, \textsf{null})$
\end{tabular}%
\caption{Shared variables}%
\label{fig:globalvars}
\end{subfigure}%
\caption{Data structures used in epoch-based algorithms, including initial values}%
\label{fig:epoch-state}
\end{figure}

In our epoch-based algorithms,
the execution of each thread is subdivided
into a sequence of discrete \emph{epochs}.
During an epoch, each thread iteratively collects samples;
the stopping condition is only checked at the end of an epoch.
The crucial advantage
of this approach is that the end of an epoch \emph{does not} require
global synchronization.
Instead, our framework guarantees the consistency of the sampled data by
maintaining multiple copies of the sampling state.

As an invariant, it is guaranteed that that
no thread writes to a copy of the state that is currently being read by another thread.
This is achieved as follows: each copy of the sampling state is labeled by
an epoch \emph{number} $e$, \ie a monotonically increasing integer
that identifies the epoch in which the data was generated.
When the stopping condition has to be checked,
all threads advance to a new epoch $e + 1$ and start writing
to a new copy of the sampling state.
The stopping condition is only verified after all threads have finished this transition
and it only takes the sampling state of epoch $e$ into account.

More precisely, the main data structure that we 
use to store the
sampling state is called a
\emph{state frame} (SF). Each SF $f$ (depicted in Figure~\ref{fig:sf})
consists of (i) an epoch number ($f.\mathtt{epoch}$),
(ii) a number of samples ($f.\mathtt{num}$) and (iii) the sampled
data ($f.\mathtt{data}$). The latter two
symbols directly correspond to $d.\texttt{num}$ and
$d.\mathtt{data}$ in our generic formulation of an adaptive sampling
algorithm (Algorithm~\ref{algo:adaptive}).
Aside from the SF structures,
our framework maintains three global variables that are shared
among all threads (depicted in Figure~\ref{fig:globalvars}):
(i) a simple Boolean flag
\texttt{stop} to determine if the algorithm should terminate,
(ii) a variable \texttt{epochToRead} that stores the number
of the epoch that we want to check the stopping condition on
and (iii) a pointer
\texttt{sfFin}[$t$] for each thread $t$ that points to a
SF finished by thread $t$.
Incrementing $\mathtt{epochToRead}$ is our synchronization
mechanism to notify all threads that they should advance
to a new epoch. This transition is visualized in Figure~\ref{fig:advance}
in Appendix~\ref{app:epoch-transition}.

\begin{algorithm}[t]
\caption{Epoch-based Approach}
\label{algo:epoch-based}
\begin{minipage}{.49\textwidth}
	Per-thread variable initialization:
	\begin{algorithmic}
		\State $e_\mathrm{sam} \gets 1$
		\State $f_\mathrm{sam} \gets$ new SF for $e_\mathrm{sam} = 1$
		\If{$t = 0$}
			\State $e_\mathrm{chk} \gets 0$
			\State $inCheck \gets$ \textsf{false}
		\EndIf
	\end{algorithmic}
	Main loop for thread $t$:
	\begin{algorithmic}[1]
		\Loop
			\State $doStop \rlxmove \mathtt{stop}$ \label{line:main-loop}
			\If{$doStop$}
				\State \textbf{break}
			\EndIf
			\State $f_\mathrm{sam}\mathtt{.data}
				\gets f_\mathrm{sam}\mathtt{.data} \circ \Call{sample}{\null}$
			\State $f_\mathrm{sam}\mathtt{.num}
				\gets f_\mathrm{sam}\mathtt{.num} + 1$
			\State $r \rlxmove$ \texttt{epochToRead} \label{line:advance-epoch}
			\If{$r =$ $e_\mathrm{sam}$}
				\State reclaim SF of epoch $e_\mathrm{sam} - 1$ \label{line:reclaim}
				\State \texttt{sfFin}[t] $\storerel$ $f_\mathrm{sam}$ \label{line:publish}
				\State $e_\mathrm{sam}$ $\gets$ $e_\mathrm{sam}$ $+ 1$
				\State $f_\mathrm{sam}$ $\gets$ new SF for $e_\mathrm{sam}$ \label{line:new-sf}
			\EndIf
			\If{$t = 0$} \label{line:thread-zero-call}
				\State \Call{checkFrames}{\null}
			\EndIf
		\EndLoop
	\algstore{epochbased}
	\end{algorithmic}
\end{minipage}\hspace{.02\textwidth}
\begin{minipage}{.49\textwidth}
	Check of stopping condition by thread $0$:
	\begin{algorithmic}[1]
	\algrestore{epochbased}
		\Procedure{checkFrames}{\null}
			\If{\textbf{not} $inCheck$} \label{line:check-cycle}
				\State $e_\mathrm{chk} \gets e_\mathrm{chk} + 1$
				\State \texttt{epochToRead} $\rlxmove e_\mathrm{chk}$
				\State $inCheck \gets$ \textsf{true}
			\EndIf
			\For{$i \in \{1, \ldots, T\}$} \label{line:check-frames}
				\State $f_\mathrm{fin} \loadacq$ \texttt{sfFin}[t] \label{line:subscribe}
				\If{$f_\mathrm{fin} = \mathsf{null}$}
					\State \Return
				\EndIf
				\If{$f_\mathrm{fin}$.\texttt{epoch} $\neq e_\mathrm{chk}$}
					\State \Return
				\EndIf
			\EndFor
			\State $d \gets$ new SF for accumulation
			\For{$i \in \{1, \ldots, T\}$} \label{line:accumulate}
				\State $f_\mathrm{fin} \rlxmove$ \texttt{sfFin}[t]
				\State $d\mathtt{.data}
					\gets d\mathtt{.data} \circ f_\mathrm{fin}\mathtt{.data}$
				\State $d\mathtt{.num}
					\gets d\mathtt{.num} + f_\mathrm{fin}\mathtt{.num}$
			\EndFor
			\If{\Call{checkForStop}{$d$}} \label{line:convergence}
				\State $\mathtt{stop} \rlxmove$ \textsf{true} \label{line:do-stop}
			\EndIf
			\State $inCheck \gets$ \textsf{false}
		\EndProcedure
	\end{algorithmic}
\end{minipage}
\end{algorithm}

Algorithm~\ref{algo:epoch-based} states the pseudocode of our framework.
By $\rlxmove$, $\loadacq$ and $\storerel$,
we denote relaxed memory access, \texttt{load-acquire} and \texttt{store-release},
respectively (see Sections~\ref{sub:basic-defs} and Appendix~\ref{app:c11_model}).
In the algorithm, each thread maintains an epoch number
$e_\mathrm{sam}$.
To be able to check the stopping condition,
thread 0 maintains another epoch number $e_\mathrm{chk}$.
Indeed, thread 0 is the only thread that evaluates the stopping condition
(in \textsc{checkFrames}) after accumulating the SFs from all threads.
\textsc{checkFrames} determines whether there is an
ongoing check for the stopping condition
($inCheck$ is \textsf{true}; line~\ref{line:check-cycle}).
If that is not the case, a check is initiated
(by incrementing $e_\mathrm{chk}$)
and all threads are signaled to advance
to the next epoch (by updating $\texttt{epochToRead}$).
Afterwards, \textsc{checkFrames} only continues
if all threads $t$ have published their SFs for checking
(\ie $\mathtt{sfFin}[t]$ points to a SF of epoch $e_\mathrm{chk}$;
line~\ref{line:check-frames}).
Once that happens, those SFs are accumulated (line~\ref{line:accumulate})
and the stopping condition is checked on
the accumulated data (line~\ref{line:convergence}).
Eventually, the termination flag (\texttt{stop}; line~\ref{line:do-stop}) signals
to all threads that they should stop sampling.
The main algorithm, on the other hand, performs a loops until this
flag is set (line~\ref{line:main-loop}). Each iteration
collects one sample and writes the results to the
current SF ($f_\mathrm{sam}$).
If a thread needs to advance to a new epoch
(because an incremented \texttt{epochToRead} is read in line~\ref{line:advance-epoch}),
it publishes its current SF to \texttt{sfFin} and starts
writing to a new SF ($f_\mathrm{sam}$; line~\ref{line:new-sf}).
Note that the memory used by old SFs can be reclaimed (line~\ref{line:reclaim};
however, note that there is no SF for epoch 0).
How exactly that is done is left to the algorithms described in later
subsections.
In the remainder of this subsection, we prove the correctness of our approach.

\begin{proposition}
Algorithm~\ref{algo:epoch-based} always checks the stopping condition
on a consistent state;
in particular, the epoch-based approach is correct.
\end{proposition}
\begin{proof}
	The order of lines~\ref{line:publish} and~\ref{line:new-sf} implies that
	no thread $t$ issues a store to a SF $f$ which it already published to
	$\mathtt{sfFin}[t]$. Nevertheless, we need to prove that all
	stores by thread $t$ are visible to \Call{checkFrames}{} before
	the frames are accumulated. \Call{checkFrames}{} only
	accumulates $f.\mathtt{data}$ after $f$ has been published to $\mathtt{sfFin}[t]$
	via the \texttt{store-relase} in line~\ref{line:publish}.
	Furthermore, in line~\ref{line:subscribe}, \Call{checkFrames}{} performs at least one
	\texttt{load-acquire} on $\mathtt{sfFin}[t]$ to read the pointer to $f$.
	Thus, all stores to $f$ are visible to \Call{checkFrames}{}
	before the accumulation in line~\ref{line:accumulate}.
	The proposition now follows from the fact that $\circ$ is associative,
	so that line~\ref{line:accumulate} indeed produces a SF that
	occurs in some sequential execution.
\qed
\end{proof}

\subsection{Local-frame and Shared-frame Algorithm}
We present two epoch-based algorithms relying on the general
framework from the previous section:
namely, the
\emph{local-frame} and the \emph{shared-frame} algorithm.
Furthermore, in Appendix~\ref{app:deterministic}, we present
the deterministic indexed-frame algorithm (as both local-frame and shared-frame
are non-deterministic).
Local-frame and shared-frame are both based on
the pseudocode in Algorithm~\ref{algo:epoch-based}. They differ, however,
in their allocation and reuse (in line~\ref{line:reclaim} of the code) of SFs.
The local frame
algorithm allocates one pair of SFs per thread and cycles through
both SFs of that pair
(\ie epochs with even numbers are assigned the first SF while odd epochs use the second SF). This
yields a per-thread memory requirement of $\mathcal{O}(n)$;
as before, $n$ denotes the size of the sampling state.
The shared-frame algorithm
reduces this memory requirement to $\mathcal{O}(1)$ by only allocating
$F$ pairs of SFs in total, for a constant number $F$. Thus, $T/F$ threads
share a SF in each epoch and atomic \texttt{fetch-add} operations need to be
used to write to the SF.
The parameter $F$ can be used to balance the memory bandwidth and synchronization
costs -- a smaller value of $F$ lowers the memory bandwidth
required during aggregation but leads to more cache contention due to atomic operations.

\subsection{Synchronization Costs}
\label{sub:sync-costs}

In Algorithm~\ref{algo:epoch-based}, all synchronization of threads $t > 0$
is done wait-free
in the sense that the threads only have to stop sampling for $\Theta(1)$
instructions to communicate with other threads
(\ie to check $\mathtt{epochToRead}$, update per-thread state and write to $\mathtt{sfFin}[t]$).
At the same time, thread $t = 0$ generally needs to check all $\mathtt{sfFin}$ pointers.
Taken together, this yields the following statement:
\begin{proposition}
	In each iteration of the main loop, threads $t > 0$
	of local-frame and shared-frame algorithms
	spend $\Theta(1)$ time to wait for other threads.
	Thread $t = 0$ spends up to $\mathcal{O}(T)$ time to wait for other threads.
\end{proposition}
In particular, the synchronization cost does not depend on the problem instance
-- this is in contrast to the OpenMP parallelization in which threads can idle
for $\mathcal{O}(\mathcal{S})$ time, where $\mathcal{S}$ denotes the time complexity
of a sampling operation (\eg $\mathcal{S} = \mathcal{O}(|V| + |E|)$ in the case
of \kad).

Nevertheless, this advantage in synchronization costs comes at a price:
the accumulation of the sampling data requires
additional evaluations of $\circ$.
$\mathcal{O}(Tn)$ evaluations are required in the local-frame algorithm,
whereas shared-frame requires $\mathcal{O}(Fn)$.
No accumulation is necessary in the OpenMP baseline.
As can be seen in Algorithm~\ref{algo:epoch-based}, we perform the
accumulation in a single thread (\ie thread 0). Compared to
a parallel implementation (\eg using parallel reductions), this strategy
requires no additional synchronization and has a favorable memory access pattern
(as the SFs are read linearly).
A disadvantage, however, is that there is a higher latency (depending on $T$) until
the algorithm detects that it is able to stop.
Appendix~\ref{app:latency} discusses how a constant latency can be achieved heuristically.


\section{Experiments}
\label{sec:experiments}
\nprounddigits{1}

\begin{figure}[t]%
\begin{subfigure}[t]{0.48\textwidth}%
	\includegraphics{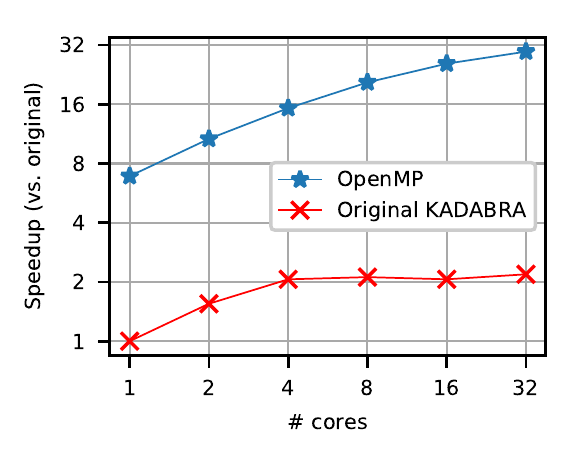}
	\caption{Average speedup (preprecessing + ADS, geom.\ mean) of OpenMP baseline over
		the original sequential implementation of \kad}
	\label{fig:original_su}
\end{subfigure}\hspace{0.04\textwidth}%
\begin{subfigure}[t]{0.48\textwidth}%
	\includegraphics{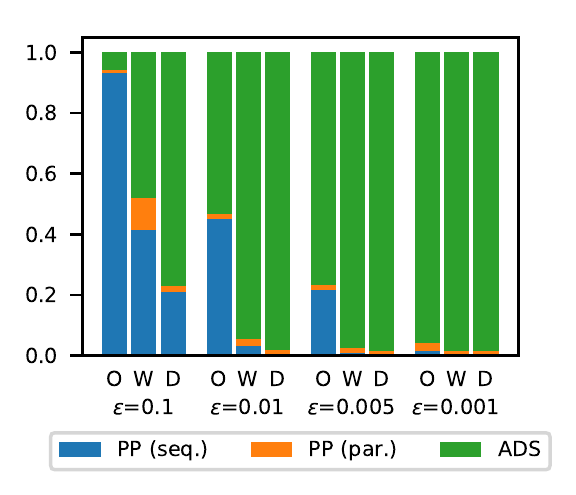}
	\caption{Breakdown of sequential \kad running times into
		preprocessing and ADS (in percent) on instances
		orkut-links~(O), wikipedia\_link\_de~(W),
		and dimacs9-COL~(D)}
	\label{fig:rt_split}
\end{subfigure}%
\caption{Performance of OpenMP baseline}
\end{figure}

\begin{figure}[t]
\begin{subfigure}[t]{0.48\textwidth}%
	\includegraphics{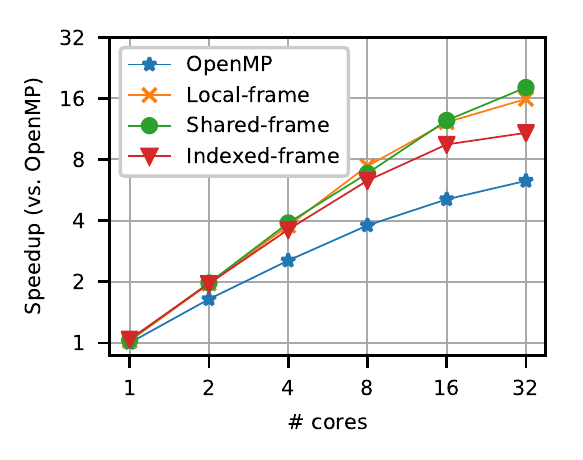}
	\caption{Average \adapt speedup (geom.\ mean) of epoch-based algorithms
		over sequential OpenMP baseline}
	\label{fig:adaptive_su}
\end{subfigure}\hspace{0.04\textwidth}%
\begin{subfigure}[t]{0.48\textwidth}%
	\includegraphics{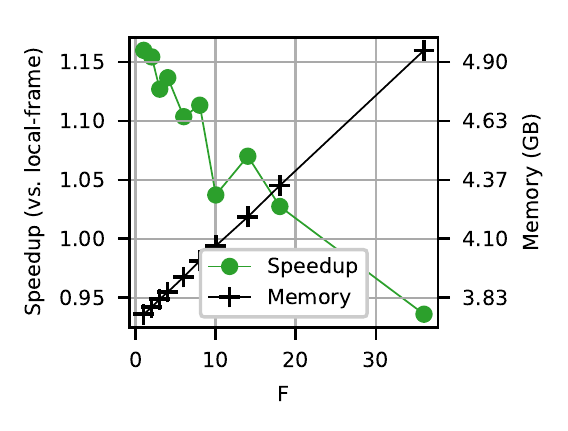}
	\caption{Average \adapt speedup (over 36-core local-frame, geom.\ mean)
		and memory consumption
		of shared-frame, depending on the number of SFs}
	\label{fig:mem_rt}
\end{subfigure}%
\caption{Performance of epoch-based algorithms}
\end{figure}

The platform we use for our experiments is a Linux server
equipped with 1.5TB RAM and two Intel Xeon Gold 6154 CPUs
with 18 cores (for a total of 36 cores) at 3.00GHz.
Each thread of the algorithm is pinned to a unique core; hyperthreading is disabled.
Our implementation\footnote{When this paper is accepted,
	we are keen to publish our code on the \nwk GitHub repository:
	\url{https://github.com/kit-parco/networkit}.
	In the meantime, our code is available at
	\url{https://gist.github.com/angriman/cfb729c1c369198b8a1a36aad1f52fcc}.}
is written in C++ building upon the
\nwk toolkit~\cite{staudt2016networkit}.
We use \numInst\xspace undirected real-world graphs in the experiments
(see Appendix~\ref{app:instances} for more details).
The error probability for \kad is set to $\delta = 0.1$ for all experiments.
Absolute running times of our experiments are reported in Appendix~\ref{app:raw-data}.

In a first experiment, we compare our OpenMP baseline
against the original implementation of \kad
(see Section~\ref{sub:first-attempts} for these two approaches).
We set the absolute approximation error to $\epsilon = 0.01$.
The overall speedup (\ie both preprocessing and ADS) is reported
in Figure~\ref{fig:original_su}.
The results show that our OpenMP baseline
outperforms the original implementation considerably
(\ie by a factor of $\numprint{\overallSuVsOriginalSingle}\times$),
even in a single-core setting.
This is mainly due to implementation tricks (see Appendix~\ref{app:impr_seq})
and parameter tuning (as discussed in Appendix~\ref{app:tuning}).
Furthermore, for 32 cores, our OpenMP baseline performs
$\numprint{\overallSuOmpVsOriginal}\times$ better than the original impelementation of \kad\ -- 
or $\numprint{\adsSuOmpVsOriginal}\times$ if only the ADS phase is considered.
Hence, for the remaining experiments, we discard the original implementation
as a competitor and focus on the parallel speedup of our algorithms.

To understand the relation between the preprocessing and ADS phases
of \kad,
we break down the running times of the OpenMP baseline in Figure~\ref{fig:rt_split}.
In this figure, we present the fraction of time that is spent in ADS
on three exemplary instances and for different values of $\epsilon$.
Especially if $\epsilon$ is small, the ADS running time dominates the overall performance of the
algorithm.
Thus, improving the scalability of the ADS phase is of critical importance.
For this reason, we neglect the preprocessing phase and only consider ADS when comparing
to our local-frame and shared-frame algorithms.

In Figure~\ref{fig:adaptive_su}, we report the parallel speedup of the ADS phase
of our epoch-based algorithms
relative to the OpenMP baseline.
All algorithms are configured to check the
stopping condition after a fixed number of samples (see Appendix~\ref{app:latency} for details).
The number $F$ of SF pairs of shared-frame has been configured to $2$,
which we found to be a good setting for $T = 32$.
On 32 cores, local-frame and shared-frame 
achieve parallel speedups of $\numprint{\adsParallelSuLf}\times$
and $\numprint{\adsParallelSuSf}$; they both significantly improve upon the OpenMP baseline,
which can only achieve a parallel speedup of $\numprint{\adsParallelSuNaive}\times$
(\ie local-frame and shared-frame are $\numprint{\adsSuLfVsOmp}\times$
	and $\numprint{\adsSuSfVsOmp}\times$ faster, respectively;
	they also outperform the original implementation by factors
	of $\numprint{\adsSuLfVsOriginal}$ and $\numprint{\adsSuSfVsOriginal}$, respectively).
The difference between local-frame and
shared-frame is insignificant for lower numbers of cores;
this is explained by the fact that the reduced memory footprint of shared-frame
only improves performance once memory bandwidth becomes a bottleneck.
For the same reason, both algorithms scale very well until 16 cores;
due to memory bandwidth limitations, this nearly ideal scalability does not extend to 32 cores.
This bandwidth issue is known to affect graph
traversal algorithms in general~\cite{bader2005architectural,lumsdaine2007challenges}.

The indexed-frame algorithm is not as fast as local-frame
and shared-frame on the
instances depicted in Figure~\ref{fig:adaptive_su}: it achieves
a parallel speedup of $\numprint{\adsParallelSuDt}\times$ on 32 cores.
However, it is still considerably faster than the OpenMP baseline
(by a factor of $\numprint{\adsSuDtVsOmp}\times$).
There are two reasons why the determinism of indexed-frame is
costly: index-frame has similar bandwidth requirements
as local-frame; however, it has to allocate more memory as SFs are buffered for
longer periods of time. On the other hand, even when enough samples are collected,
the stopping condition has to be checked on older samples first, while
local-frame and shared-frame can just check the stopping condition on the
most recent sampling state.

In a final experiment, we evaluate the impact of the parameter $F$
of shared-frame on its performance. Figure~\ref{fig:mem_rt} depicts
the results.
The experiment is done with 36 cores; hence memory pressure is even
higher than in the previous experiments.
The figure demonstrates that in this situation,
minimizing the memory bandwith requirements at the expense
of synchronization overhead is a good strategy.
Hence for larger numbers of cores,
we can minimize memory footprint and maximize performance at the same time.

\npnoround

\section{Related Work}
\label{sec:related_work}
Our parallelization strategy can be applied to arbitrary
\adapt algorithms.
\adapt was first introduced by Lipton and Naughton to
estimate the size of the transitive closure of a digraph~\cite{lipton1989estimating}.
It is used in a variety of fields, \eg in statistical learning~\cite{provost1999efficient}.
In the context of \bc, \adapt has been used
to approximate distances between
pairs of vertices of a graph~\cite{oktay2011distance},
to approximate the \bc values of
a graph~\cite{bader2007approximating,riondato2018abra,borassi2016kadabra}
and to approximate the \bc value of a single vertex~\cite{chehreghani2018novel}.
An analogous strategy is exploited by Mumtaz and Wang~\cite{mumtaz2017identifying} to
find approximate solutions to the group betweenness maximization problem.

Regarding more general (\ie not necessarily \adapt) algorithms for \bc,
a survey from Matta \etal~\cite{matta2019comparing} provides a detailed
overview of the the state of the art.
The \rk~\cite{riondato2016fast} algorithm represents the leading non-adaptive
sampling algorithm for \bc approximation;
\kad was shown to be 100 times faster than \rk in undirected real-world graphs,
and 70 times faster
than \rk in directed graphs\cite{borassi2016kadabra}.
McLaughlin and Bader~\cite{mclaughlin2014scalable} introduced a work-efficient parallel
algorithm for \bc approximation, implemented for single- and multi-GPU machines.
Madduri \etal~\cite{madduri2009faster} presented a lock-free parallel algorithm
optimized for specific architectures to approximate or compute \bc exactly
in massive networks.

The SFs used by our algorithms are concurrent data structures that enable us to
minimize the synchronization latencies in multithread environments.
Devising concurrent (lock-free) data structures that scale over multiple cores 
is not trivial and much effort has been devoted to this goal~\cite{boyd2010analysis,michael2004hazard}.
A well-known solution is the Read-Copy-Update mechanism (RCU); it 
was introduced to achieve high multicore scalability on read-mostly data structures~\cite{mckenney1998read},
and was leveraged by several applications~\cite{arbel2014concurrent,clements2012scalable}.
Concurrent hash tables~\cite{DBLP:conf/sosp/DavidGT13} are another popular example.


\section{Conclusions and Future Work}
\label{sec:concl}
In this paper, we found that previous techniques to parallelize ADS algorithms
are insufficient to scale to large numbers of threads. However, significant speedups
can be achieved by employing adequate concurrent data structures.
Using such data structures and our epoch mechanism,
we were able to devise parallel ADS algorithms
that consistently outperform the state of the art
but also achieve different trade-offs between synchronization costs,
memory footprint and determinism of the results.

Regarding future work, a promising direction for our algorithms
is parallel computing with distributed memory; here, the stopping
condition could be checked via (asynchronous) reduction of the SFs.
In the case of \bc this, might yield a way to avoid
bottlenecks for memory bandwidth on shared-memory systems.


\bibliographystyle{splncs04}
\bibliography{biblio}

\newpage
\appendix
\section{The C11 memory model}
\label{app:c11_model}

As mentioned in Section~\ref{sub:basic-defs}, we work in the C11 memory model.
The weakest operations in this model are \texttt{load-relaxed} and \texttt{store-relaxed}
operations; those only guarantee the atomicity of the memory access
(\ie they guarantee that no \emph{tearing} occurs)
but no ordering at all. Hence, the order in which \texttt{store-relaxed}
writes become visible to \texttt{load-relaxed} reads can differ from the
order in which the stores and loads are performed by individual threads.
\texttt{load-acquire} and \texttt{store-release} additionally
do provide ordering guarantees:
if thread $t$ writes a word $X$ to a given memory location using \texttt{store-release}
\emph{and} thread $t'$ reads $X$ using \texttt{load-acquire} from the same memory location,
then all store operations -- whether atomic or not --
done by thread $t$ \emph{before} the store of $X$
become visible to all load operations done by thread $t'$ \emph{after} the load of $X$.
We note that C11 defines even stronger ordering guarantees that we do not require in
this paper.
Furthermore, on a hardware level, x86\_64 implements a stronger
\emph{total store order};
thus, \texttt{load-acquire} and \texttt{store-release} compile to plain
load and store instructions and our local-frame algorithm does not perform \emph{any}
synchronization instructions on x86\_64.


\section{Details of \kad Algorithm}
\label{app:kadabra}
In Section~\ref{sub:kadabra_algo} we described the \kad algorithm;
in this appendix, we illustrate the stopping condition more in detail and
show that evaluating it in a consistent state is crucial for the correctness
of the algorithm.
The functions $f$ and $g$ we mentioned in Eq.~\ref{eq:kad-inv} are
defined as~\cite{borassi2016kadabra}:
\begin{align*}
  f(\approxbc(v), \delta_{L}(v), \omega, \tau) &=
  \frac{1}{\tau}\left(\log\frac{1}{\delta_{L}(v)}\right)\left(
  \frac{1}{3} - \frac{\omega}{\tau}
  + \sqrt{\left(\frac{1}{3} - \frac{\omega}{\tau} \right)^2 +
  \frac{2\approxbc(v)\omega}{\log\frac{1}{\delta_{L}(v)}}}
  \right)\\
  g(\approxbc(v), \delta_{U}(v), \omega, \tau) &=
  \frac{1}{\tau}\left(\log\frac{1}{\delta_{U}(v)}\right)\left(
  \frac{1}{3} + \frac{\omega}{\tau}
  + \sqrt{\left(\frac{1}{3} + \frac{\omega}{\tau} \right)^2 +
  \frac{2\approxbc(v)\omega}{\log\frac{1}{\delta_{U}(v)}}}
  \right)
\end{align*}

where $\approxbc(v)$ is the approximation of the \bc of vertex $v$ obtained after
$\tau$ samples.
When the stopping condition is evaluated, $f$ and $g$ are computed for every
vertex of the graph and the algorithm terminates if:
\[
f(\approxbc(v), \delta_{L}(v), \omega, \tau) \le \epsilon
\quad \text{and} \quad
g(\approxbc(v), \delta_{U}(v), \omega, \tau) \le \epsilon
\]

holds for every vertex $v$ of the graph.
It is straightforward to verify that both $f$ and $g$ grow with $\approxbc(v)$
but that $g$ decreases with $\tau$.
Thus, evaluating the stopping condition with inconsistent data
(\eg if accesses to $\tau$ and $\approxbc(v)$ are not synchronized) could lead to an
erroneous termination of the algorithm.


\section{Optimization and Tuning}

\subsection{Improvements to the \kad Implementation}
\label{app:impr_seq}
In the following we document some improvements to the sequential \kad implementation
of Borassi and Natale~\cite{borassi2016kadabra}.
First, we avoid searching for non-existing shortest paths
between a pair $(s, t)$ of selected vertices
by checking if $s$ and $t$ belong to the same connect
component\footnote{Connected
	components are computed along with the diameter during preprocessing.}.
Then, we reduce the memory footprint of the sampling procedure:
the original \kad implementation
stores all predecessors on shortest paths in a separate graph $G'$, which is
used to backtrack the path starting from the last explored vertices.
Our implementation avoids the use of $G'$ by
reconstructing shortest $s$-$t$-paths from the original graph $G$
and a distance vector.
Furthermore, for each shortest $s$-$t$-path sampled,
the original \kad implementation needs to reset a Boolean
\enquote{visited} array
with an overall additional cost of
$\Theta(|V|)$ time per sample.
We avoid doing this by using 7 bits per element in this array to store a 
\emph{timestamp} that indicates when the vertex was last visited; therefore,
the array needs to be reset only once in $2^7 = 128$ BFSs.

\subsection{Balancing Costs of Termination Checks}
\label{app:tuning}
Although the pseudocode of Algorithms~\ref{algo:adaptive} and~\ref{algo:epoch-based}
checks the stopping condition
after every sample, this amount of checking is excessive in practice.
Hence, both the original \kad and the OpenMP ADS algorithms
check the stopping condition after a fixed number $N$ of samples.
$N$ represents a tradeoff between the time required to check the
stopping condition and the time required to sample a shortest path.
In the original \kad implementation, $N$ is set to 11;
however, this choice turned out to be inefficient in our experiments.
Thus, we formed a small set of the instances for parameter
tuning\footnote{We chose the instances com-amazon, munmun\_twitter\_social,
orkut-links, roadNet-PA, wikipedia\_link\_de, and wikipedia\_link\_fr.}
and ran experiments with different values of $N$. 
As the result, we found that $N = 1000$ empirically performs best.

\subsection{Termination Latency in Epoch-based Approach}
\label{app:latency}
In the epoch-based approach, we also need to balance the frequency
of checking the stopping condition and the time invested into sampling; however,
we face a different problem: the accumulation of all SFs before
the stopping condition is checked takes $\Oh(Tn)$ time, thus the length
of an epoch depends on $T$ (see Section~\ref{sub:sync-costs}).
This is an undesirable artifact as it introduces an
additional delay between the time when the algorithm could potentially stop
(because enough samples have been collected)
and the time when the algorithm actually stops
(because the accumulation is completed).
It would be preferable to check the stopping condition after
a constant number of samples (summed over all threads)
-- as the sequential and OpenMP variants
naturally do.\footnote{As a side effect, doing so
	improves the comparability of those algorithms.}

While it seems unlikely that a constant number of samples per epoch can
be achieved (without additional synchronization overhead),
we aim to satisfy this property heuristically.
Checking the stopping condition after $N_0 = (1/T) N$
samples \emph{per thread} seems to be a reasonable heuristic.
However, it does not account for the fact that only one thread performs
the check while all additional threads continue to sample data.
Thus, we check the stopping condition after
\[
  N_0 = \frac1{T^\xi} N
\]
samples from thread 0.
Here, $\xi$ is another parameter that can be tuned.
Using the same approach as in Appendix~\ref{app:tuning}
(and running the algorithm on 32 cores),
we empirically determined $\xi = \log_{32}(N/10) \approx 1.33$
to be a good choice.


\section{Details on Epoch-Based Algorithms}

\subsection{Visualization of Epoch Transition}
\label{app:epoch-transition}

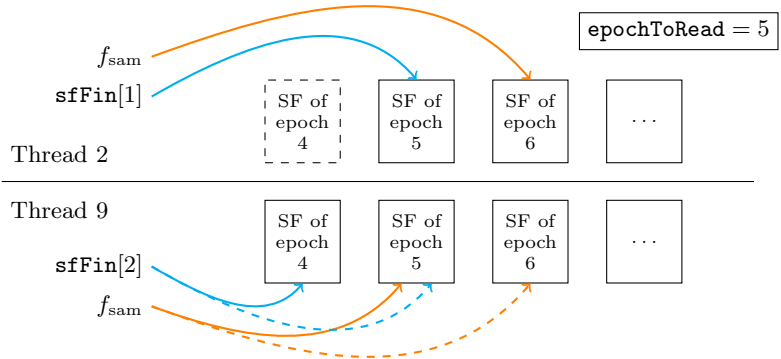
\begin{figure}[t]
\centering
\begin{tikzpicture}[sfstyle/.style={
	draw,rectangle,minimum width=1cm,minimum height=1.1cm,align=center,font=\scriptsize
}]
	\clip (-4,-2.4) rectangle (7,2.4);
	\node[draw,rectangle] at (5,2) {$\mathtt{epochToRead} = 5$};

	\draw (-4,0) -- (6,0);
	\node[above right=0.15cm and 0cm] at (-4,0) {Thread 2};
	\node[below right=0.15cm and 0cm] at (-4,0) {Thread 9};

	\node[sfstyle,dashed] at (0,0.8) (sf00) {SF of\\epoch\\4};
	\node[sfstyle,right=0.5cm of sf00] (sf01) {SF of\\epoch\\5};
	\node[sfstyle,right=0.5cm of sf01] (sf02) {SF of\\epoch\\6};
	\node[sfstyle,right=0.5cm of sf02] (sf03) {$\ldots$};
	\node[text width=1.8cm,align=right,above left=-0.5cm and 1.5cm of sf00] (efin0) {$\mathtt{sfFin}[1]$};
	\node[text width=1.8cm,align=right,above=0cm of efin0] (csf0) {$f_\mathrm{sam}$};
	\draw[->,cyan,thick] (efin0.east) to[out=30,in=130] (sf01.north);
	\draw[->,orange,thick] (csf0.east) to[out=20,in=130] (sf02.north);

	\node[sfstyle] at (0,-0.8) (sf10) {SF of\\epoch\\4};
	\node[sfstyle,right=0.5cm of sf10] (sf11) {SF of\\epoch\\5};
	\node[sfstyle,right=0.5cm of sf11] (sf12) {SF of\\epoch\\6};
	\node[sfstyle,right=0.5cm of sf12] (sf13) {$\ldots$};
	\node[text width=1.8cm,align=right,below left=-0.5cm and 1.5cm of sf10] (efin1) {$\mathtt{sfFin}[2]$};
	\node[text width=1.8cm,align=right,below=0cm of efin1] (csf1) {$f_\mathrm{sam}$};
	\draw[->,cyan,thick] (efin1.east) to[out=-30,in=-130] (sf10.south);
	\draw[->,orange,thick] (csf1.east) to[out=-20,in=-130] ([xshift=-0.2cm]sf11.south);
	\draw[->,cyan,dashed,thick] (efin1.east) to[out=-30,in=-130] ([xshift=0.2cm]sf11.south);
	\draw[->,orange,dashed,thick] (csf1.east) to[out=-20,in=-130] (sf12.south);
\end{tikzpicture}
\caption{Transition after $\mathtt{epochToRead}$ is set to $5$. Thread 2 already
	writes to the SF of epoch~6 (using the $f_\mathrm{sam}$ pointer).
	Thread 9 still writes to the SF of epoch~5 but
	advances to epoch~6 once it checks $\mathtt{epochToRead}$ (dashed orange line).
	Afterwards, thread 9 publishes its SF of epoch~5 to $\mathtt{sfFin}$ (dashed blue line).
	Finally, the stopping condition is checked using both SFs of epoch~5
	(\ie the SFs now pointed to by $\mathtt{sfFin}$).}
\label{fig:advance}
\end{figure}

Figure~\ref{fig:advance} visualizes the update of the $\mathtt{sfFin}$ pointers
after an epoch transition is initiated by incrementing $\mathtt{epochToRead}$.
The exact mechanism is discussed in Section~\ref{sub:epoch_based}


\subsection{Indexed-frame Algorithm}
\label{app:deterministic}

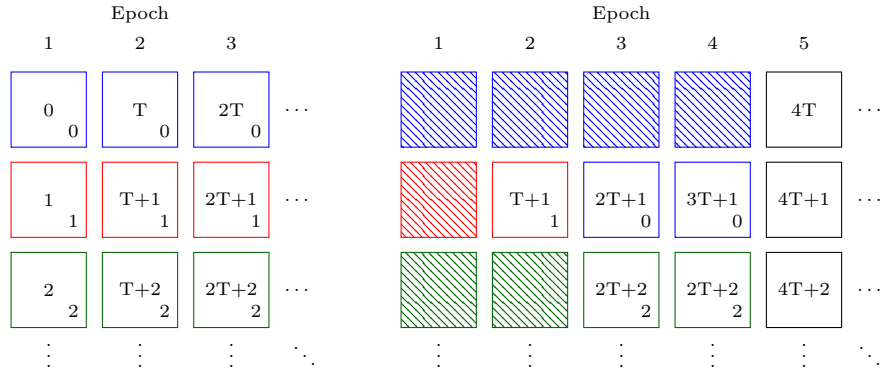
\begin{figure}[t]
\centering
\begin{subfigure}[t]{0.375\textwidth}
\scriptsize
  \begin{tikzpicture}[sfstyle/.style={
	draw,rectangle,minimum width=1cm,minimum height=1cm,align=center,font=\scriptsize
}]
	\node[sfstyle,draw=blue] at (0,0) (t0) {0};
	\node[sfstyle,right=0.2cm of t0,draw=blue] (t1) {T};
	\node[sfstyle,right=0.2cm of t1,draw=blue] (t2) {2T};
	\node[right=0.1cm of t2] (td2) {\dots};

	\node[sfstyle,draw=red] at (0,-1.2) (t10) {1};
	\node[sfstyle,right=0.2cm of t10,draw=red] (t11) {T+1};
	\node[sfstyle,right=0.2cm of t11,draw=red] (t12) {2T+1};
	\node[right=0.1cm of t12] (td3) {\dots};
	
	\node[sfstyle,draw=\darkgreen] at (0, -2.4) (t20) {2};
	\node[sfstyle,right=0.2cm of t20,draw=\darkgreen] (t21) {T+2};
	\node[sfstyle,right=0.2cm of t21,draw=\darkgreen] (t22) {2T+2};
	\node[right=0.1cm of t22] (td4) {\dots};
	
	\node[below=-0.1cm of t20] (td1) {\vdots};
	\node[below=-0.1cm of t21] (td11) {\vdots};
	\node[below=-0.1cm of t22] (td12) {\vdots};
	\node[below right=-0.1cm and 0.2cm of t22] (tdd) {$\ddots$};
	
	\node[above=0.55cm of t1] (nEp) {Epoch};
	
	\node[below right=-0.4cm and -0.35cm of t0] (th0) {0};
	\node[below right=-0.4cm and -0.35cm of t1] (th0) {0};
	\node[below right=-0.4cm and -0.35cm of t2] (th0) {0};
	
	\node[below right=-0.4cm and -0.35cm of t10] (th0) {1};
	\node[below right=-0.4cm and -0.35cm of t11] (th0) {1};
	\node[below right=-0.4cm and -0.35cm of t12] (th0) {1};
	
    \node[below right=-0.4cm and -0.35cm of t20] (th0) {2};
	\node[below right=-0.4cm and -0.35cm of t21] (th0) {2};
	\node[below right=-0.4cm and -0.35cm of t22] (th0) {2};
	
	\node[above=0.2cm of t0] (e1) {1};
	\node[above=0.2cm of t1] (e2) {2};
	\node[above=0.2cm of t2] (e3) {3};
\end{tikzpicture}
\caption{SF indices in indexed-frame algorithm (not to be confused with epoch numbers).}
\label{fig:threads_epochs}
\end{subfigure}%
\hspace{0.05\textwidth}%
\begin{subfigure}[t]{0.525\textwidth}
\scriptsize
\begin{tikzpicture}[sfstyle/.style={
	draw,rectangle,minimum width=1cm,minimum height=1cm,align=center
}]

	\node[sfstyle,pattern=north west lines, pattern color=blue,draw=blue] at (0,0) (t0) {};
	\node[sfstyle,right=0.2cm of t0,pattern=north west lines, pattern color=blue,draw=blue] (t1) {};
	\node[sfstyle,right=0.2cm of t1,pattern=north west lines, pattern color=blue,draw=blue] (t2) {};
	\node[sfstyle,right=0.2cm of t2,pattern=north west lines, pattern color=blue,draw=blue] (t3) {};
	\node[sfstyle,right=0.2cm of t3] (t4) {4T};
	\node[right=0.1cm of t4] (td2) {\dots};
	\node[above=0.2cm of t0] (e0) {1};
	\node[above=0.2cm of t1] (e1) {2};
	\node[above=0.2cm of t2] (e2) {3};
	\node[above=0.2cm of t3] (e3) {4};
	\node[above=0.2cm of t4] (e4) {5};
	\node[above=.55cm of t2] (tEpoch) {Epoch};

	\node[sfstyle,pattern=north west lines, pattern color=red,draw=red] at (0,-1.2) (t10) {};
	\node[sfstyle,right=0.2cm of t10,draw=red] (t11) {T+1};
	\node[below right=-0.4cm and -0.35cm of t11] (th11) {1};
	\node[sfstyle,right=0.2cm of t11,draw=blue] (t12) {2T+1};
	\node[below right=-0.4cm and -0.35cm of t12] (th12) {0};
	\node[sfstyle,right=0.2cm of t12,draw=blue] (t13) {3T+1};
	\node[below right=-0.4cm and -0.35cm of t13] (th13) {0};
	\node[sfstyle,right=0.2cm of t13] (t14) {4T+1};
	\node[right=0.1cm of t14] (td22) {\dots};
	
	\node[sfstyle,pattern=north west lines, pattern color=\darkgreen,draw=\darkgreen] at (0, -2.4) (t20) {};
	\node[sfstyle,right=0.2cm of t20,pattern=north west lines, pattern color=\darkgreen,draw=\darkgreen] (t21) {};
	\node[sfstyle,right=0.2 of t21,draw=\darkgreen] (t22) {2T+2};
	\node[sfstyle,right=0.2 of t22,draw=\darkgreen] (t23) {2T+2};
	\node[sfstyle,right=0.2 of t23] (t24) {4T+2};
	\node[below right=-0.4cm and -0.35cm of t22] {2};
	\node[below right=-0.4cm and -0.35cm of t23] {2};
	\node[right=0.1cm of t24] (td5) {\dots};
	\node[below=-0.1cm of t20] (td1) {\vdots};
	\node[below=-0.1cm of t21] (td21) {\vdots};
	\node[below=-0.1cm of t22] (td22) {\vdots};
	\node[below=-0.1cm of t23] (td23) {\vdots};
	\node[below=-0.1cm of t24] (td24) {\vdots};
	\node[below right=-0.1cm and 0.1cm of t24] (td25) {$\ddots$};
\end{tikzpicture}
\caption{Reservation of SFs for $a = 2$.
	Thread 0 (blue) has
	finished SFs with indices $T$, $2T$, $3T$ and $4T$.
	Because thread 1 (red)
	did not finish SF $T+1$ yet,
	thread 0 reserves indices $2T+1$ and $3T+1$.}
\label{fig:opt_epochs}
\end{subfigure}
\caption{Indices of SFs in indexed-frame algorithm.
	Central numbers indicate SF indices.
	Numbers in bottom right corners (and colors) denote the thread
	that will compute the SF.
	Dashed SFs are already finished.}
\label{fig:grid}
\end{figure}

In this subsection, we introduce the \emph{indexed-frame} algorithm
that is a variant of local-frame but always obtains deterministic results.
In particular, we highlight the modifications compared to
local-frame that are necessary to avoid non-determinism.

There are two sources of non-determinism in the epoch-based algorithms:
First, because threads generate random numbers
independently from each other and the pseudo-random number generator (PRNG)
of each thread is seeded
differently, the sequence of generated random numbers depends on the
number of threads.
Secondly, and more importantly, the point in time where a thread notices
that the stopping condition needs to be checked
(\ie $\mathtt{epochToRead}$ is read in line~\ref{line:advance-epoch}
of Algorithm~\ref{algo:epoch-based})
is non-deterministic.
Thus, among multiple executions of the algorithm,
the SFs that are checked differ
in the number of samples and in the PRNG state
used to generate the samples.

Indexed-frame avoids the first problem
by re-seeding the random number generator
of each thread whenever the thread moves to a new epoch.
To avoid a dependence on the number of threads,
the new seed should only vary based on a unique \emph{index} of the
generated SF (not to be confused with the epoch number).
As an index for the SF of epoch $e$, we choose
$(e T + t)$,
as every thread $t$ contributes exactly one SF to each epoch $e$.
This scheme is depicted in Figure~\ref{fig:threads_epochs}.

Handling the second issue turns out to be more involved.
As we need to ensure that the stopping condition is always checked on exactly
the same SFs, the point in time where a thread
moves to a new epoch must be independent of the time when
the stopping condition is checked. To achieve that,
indexed-frame writes a fixed number of samples to each SF.
That, however, means that
by the time a check is performed, a thread can have
finished multiple SFs.
To deal with multiple finished
SFs, we use
a per-thread queue of SFs which have already been finished
but which were not considered by the stopping condition yet.
While the size of this queue is unbounded in theory,
in our experiments we never observed a thread buffering more than
12 SFs at a time (with an average of 3 SFs allocated per thread); thus, we do not implement
a sophisticated strategy to bound the queue length.
The following subsection, however, discusses such a strategy
for ADS algorithms where this becomes a problem.

\subsection{Bounded Memory Complexity in Indexed-frame}
As our experiments demonstrated, the SF buffering overhead
of the deterministic algorithm is not problematic in practice.
However, at the cost of additional synchronization,
it is possible to also bound the theoretical memory
complexity the algorithm. In particular,
if there are lower and upper bounds $\mathcal{C}_\ell$ and $\mathcal{C}_u$
on the time to compute a single SF,
we can \emph{reserve} SF indices to bound the number of simultaneously allocated
SFs\footnote{Such bounds trivially exist if the algorithmic complexity
	of a single sampling operation is bounded.}:
instead of computing the SFs with indices $(t, t + T, t + 2T, \ldots)$,
each thread $t$ determines (by synchronizing with all other threads) the
smallest index $i$ of a SF that is not reserved yet
(see Figure~\ref{fig:opt_epochs} as an illustration of this process).
Then, the SFs with indices $(i, i + T, i + 2T, \ldots, i + aT)$
for some constant $a$ are reserved by $t$ and $t$ computes
exactly those SFs before doing another reservation.
The bounds on the computation time of a single SF imply that
all other threads can only perform a constant number
$\mathcal{C}_u/\mathcal{C}_\ell$
of reservations
until an epoch is finished (and all SFs of the epoch can be reclaimed).
The constant $a$ can be chosen to balance the maximal number of buffered SFs
and synchronization costs required for reservation.


\newpage
\section{List of Instances}
\label{app:instances}

\begin{table}
\centering
\label{table:input_instances}
\caption{List of instances used for the experiments.}
\bigskip
\scriptsize
\begin{tabular}{lrrrr}
Network name & \# of vertices & \# of edges & Diameter & Category \\
\midrule
\midrule
tntp-ChicagoRegional & \numprint{12979} &
			\numprint{20627} & \numprint{106} & Infrastructure \\
dimacs9-NY & \numprint{264346} &
			\numprint{365050} & \numprint{720} & Infrastructure \\
dimacs9-COL & \numprint{435666} &
			\numprint{521200} & \numprint{1255} & Infrastructure \\
munmun\_twitter\_social & \numprint{465017} &
			\numprint{833540} & \numprint{8} & Social \\
com-amazon & \numprint{334863} &
			\numprint{925872} & \numprint{47} & Co-purchase \\
loc-gowalla\_edges & \numprint{196591} &
			\numprint{950327} & \numprint{16} & Social \\
web-NotreDame & \numprint{325729} &
			\numprint{1090108} & \numprint{46} & Hyperlink \\
roadNet-PA & \numprint{1088092} &
			\numprint{1541898} & \numprint{794} & Infrastructure \\
roadNet-TX & \numprint{1379917} &
			\numprint{1921660} & \numprint{1064} & Infrastructure \\
web-Stanford & \numprint{281903} &
			\numprint{1992636} & \numprint{753} & Hyperlink \\
petster-dog-household & \numprint{256127} &
			\numprint{2148179} & \numprint{11} & Social \\
flixster & \numprint{2523386} &
			\numprint{7918801} & \numprint{8} & Social \\
as-skitter & \numprint{1696415} &
			\numprint{11095298} & \numprint{31} & Computer \\
dbpedia-all & \numprint{3966895} &
			\numprint{12610982} & \numprint{146} & Relationship \\
actor-collaboration & \numprint{382219} &
			\numprint{15038083} & \numprint{13} & Collaboration \\
soc-pokec-relationships & \numprint{1632803} &
			\numprint{22301964} & \numprint{14} & Social \\
soc-LiveJournal1 & \numprint{4846609} &
			\numprint{42851237} & \numprint{20} & Social \\
livejournal-links & \numprint{5204175} &
			\numprint{48709621} & \numprint{23} & Social \\
wikipedia\_link\_ceb & \numprint{7891015} &
			\numprint{63915385} & \numprint{9} & Hyperlink \\
wikipedia\_link\_ru & \numprint{3370462} &
			\numprint{71950918} & \numprint{10} & Hyperlink \\
wikipedia\_link\_sh & \numprint{3924218} &
			\numprint{76439386} & \numprint{9} & Hyperlink \\
wikipedia\_link\_de & \numprint{3603726} &
			\numprint{77546982} & \numprint{14} & Hyperlink \\
wikipedia\_link\_it & \numprint{2148791} &
			\numprint{77875131} & \numprint{9} & Hyperlink \\
wikipedia\_link\_sv & \numprint{6100692} &
			\numprint{99864874} & \numprint{10} & Hyperlink \\
wikipedia\_link\_fr & \numprint{3333397} &
			\numprint{100461905} & \numprint{10} & Hyperlink \\
wikipedia\_link\_sr & \numprint{3175009} &
			\numprint{103310837} & \numprint{10} & Hyperlink \\
orkut-links & \numprint{3072441} &
			\numprint{117184899} & \numprint{10} & Social \\
\end{tabular}

\end{table}


\newpage
\section{Detailed Experimental Data}
\label{app:raw-data}

In this appendix we show the detailed running time of our algorithms on every
instance.
For better readability, we partitioned the instances into two categories:
moderate instances achieved a total running time of less than
100 seconds (Table~\ref{tab:moderate_rt});
the others are shown in Table~\ref{tab:expensive_rt}.

\begin{table}
\caption{Absolute running times (s) on moderate instances.
	Total: \adapt with preprocessing on a single core.}
\label{tab:moderate_rt}
\bigskip

\begin{subtable}{\textwidth}
	\caption{OMP: OpenMP baseline,
		L: local-frame}
	\label{tab:sub_rt11}
	\bigskip

	\centering
	\scriptsize
	\begin{tabular}{lr | rr | rr | rr | rr | rr | rr}
Network Name &  & \multicolumn{2}{c|}{1 core} & \multicolumn{2}{c|}{2 cores} & \multicolumn{2}{c|}{4 cores} & \multicolumn{2}{c|}{8 cores} & \multicolumn{2}{c|}{16 cores} & \multicolumn{2}{c}{32 cores}\\
 & Total & OMP & L & OMP & L & OMP & L & OMP & L & OMP & L & OMP & L\\
\midrule\midrule
tntp-ChicagoRegional & \numprint{6.70} & \numprint{6.62} & \numprint{5.66} & \numprint{3.25} & \numprint{2.83} & \numprint{1.56} & \numprint{1.37} & \numprint{0.85} & \numprint{0.66} & \numprint{0.45} & \numprint{0.33} & \numprint{0.27} & \numprint{0.16}\\
munmun\_twitter\_social & \numprint{7.99} & \numprint{1.72} & \numprint{1.49} & \numprint{1.41} & \numprint{0.83} & \numprint{1.09} & \numprint{0.45} & \numprint{0.89} & \numprint{0.24} & \numprint{0.84} & \numprint{0.23} & \numprint{0.78} & \numprint{0.17}\\
com-amazon & \numprint{10.49} & \numprint{9.47} & \numprint{9.18} & \numprint{4.47} & \numprint{4.38} & \numprint{3.02} & \numprint{2.35} & \numprint{2.27} & \numprint{1.34} & \numprint{1.94} & \numprint{0.86} & \numprint{1.41} & \numprint{0.54}\\
loc-gowalla\_edges & \numprint{2.82} & \numprint{2.50} & \numprint{2.09} & \numprint{1.49} & \numprint{0.99} & \numprint{1.11} & \numprint{0.49} & \numprint{0.87} & \numprint{0.20} & \numprint{0.70} & \numprint{0.11} & \numprint{0.67} & \numprint{0.10}\\
web-NotreDame & \numprint{7.66} & \numprint{7.33} & \numprint{6.55} & \numprint{4.34} & \numprint{3.30} & \numprint{3.17} & \numprint{1.72} & \numprint{2.50} & \numprint{0.68} & \numprint{2.14} & \numprint{0.43} & \numprint{1.93} & \numprint{0.33}\\
web-Stanford & \numprint{34.62} & \numprint{33.87} & \numprint{29.95} & \numprint{15.76} & \numprint{15.54} & \numprint{11.62} & \numprint{7.95} & \numprint{7.96} & \numprint{2.79} & \numprint{5.49} & \numprint{1.75} & \numprint{4.48} & \numprint{1.33}\\
petster-dog-household & \numprint{5.31} & \numprint{4.83} & \numprint{3.89} & \numprint{2.67} & \numprint{2.12} & \numprint{1.82} & \numprint{1.09} & \numprint{1.43} & \numprint{0.67} & \numprint{1.30} & \numprint{0.56} & \numprint{1.32} & \numprint{0.42}\\
flixster & \numprint{13.99} & \numprint{10.94} & \numprint{10.03} & \numprint{7.87} & \numprint{5.77} & \numprint{6.61} & \numprint{3.20} & \numprint{5.49} & \numprint{1.91} & \numprint{4.90} & \numprint{1.32} & \numprint{4.78} & \numprint{1.32}\\
as-skitter & \numprint{17.14} & \numprint{13.76} & \numprint{13.16} & \numprint{9.85} & \numprint{7.57} & \numprint{7.33} & \numprint{3.99} & \numprint{5.80} & \numprint{2.55} & \numprint{5.14} & \numprint{1.77} & \numprint{5.11} & \numprint{2.21}\\
actor-collaboration & \numprint{8.69} & \numprint{5.87} & \numprint{6.21} & \numprint{3.88} & \numprint{3.18} & \numprint{2.60} & \numprint{1.96} & \numprint{1.82} & \numprint{1.16} & \numprint{1.41} & \numprint{0.68} & \numprint{1.09} & \numprint{0.54}\\
soc-pokec-relationships & \numprint{25.38} & \numprint{16.57} & \numprint{18.21} & \numprint{10.37} & \numprint{9.00} & \numprint{8.00} & \numprint{5.23} & \numprint{6.07} & \numprint{3.02} & \numprint{5.28} & \numprint{2.56} & \numprint{5.40} & \numprint{2.08}\\
soc-LiveJournal1 & \numprint{54.91} & \numprint{36.52} & \numprint{39.08} & \numprint{31.53} & \numprint{22.37} & \numprint{22.29} & \numprint{11.68} & \numprint{17.79} & \numprint{6.12} & \numprint{15.57} & \numprint{4.69} & \numprint{14.82} & \numprint{4.03}\\
livejournal-links & \numprint{62.27} & \numprint{46.19} & \numprint{44.52} & \numprint{31.16} & \numprint{24.99} & \numprint{23.49} & \numprint{13.43} & \numprint{18.11} & \numprint{7.57} & \numprint{15.46} & \numprint{4.90} & \numprint{15.51} & \numprint{4.33}\\
wikipedia\_link\_sh & \numprint{41.54} & \numprint{21.68} & \numprint{17.98} & \numprint{17.43} & \numprint{9.49} & \numprint{14.74} & \numprint{4.68} & \numprint{13.13} & \numprint{2.44} & \numprint{12.36} & \numprint{2.05} & \numprint{12.08} & \numprint{2.11}\\
wikipedia\_link\_sr & \numprint{56.30} & \numprint{45.55} & \numprint{42.66} & \numprint{32.21} & \numprint{21.83} & \numprint{20.28} & \numprint{10.69} & \numprint{15.63} & \numprint{6.08} & \numprint{12.92} & \numprint{3.72} & \numprint{12.73} & \numprint{2.69}\\
\end{tabular}

\end{subtable}
\begin{subtable}{\textwidth}
	\caption{S: shared-frame,
		I: indexed-frame}
	\label{tab:sub_rt21}
	\bigskip

	\centering
	\scriptsize
	\begin{tabular}{lr | rr | rr | rr | rr | rr | rr}
Network Name &  & \multicolumn{2}{c|}{1 core} & \multicolumn{2}{c|}{2 cores} & \multicolumn{2}{c|}{4 cores} & \multicolumn{2}{c|}{8 cores} & \multicolumn{2}{c|}{16 cores} & \multicolumn{2}{c}{32 cores}\\
 & Total & S & I & S & I & S & I & S & I & S & I & S & I\\
\midrule\midrule
tntp-ChicagoRegional & \numprint{6.70} & \numprint{6.71} & \numprint{5.48} & \numprint{3.30} & \numprint{2.75} & \numprint{1.48} & \numprint{1.38} & \numprint{0.64} & \numprint{0.70} & \numprint{0.31} & \numprint{0.43} & \numprint{0.15} & \numprint{0.29}\\
munmun\_twitter\_social & \numprint{7.99} & \numprint{1.51} & \numprint{1.60} & \numprint{0.80} & \numprint{0.90} & \numprint{0.45} & \numprint{0.49} & \numprint{0.28} & \numprint{0.29} & \numprint{0.20} & \numprint{0.17} & \numprint{0.19} & \numprint{0.23}\\
com-amazon & \numprint{10.49} & \numprint{8.44} & \numprint{8.88} & \numprint{4.04} & \numprint{4.21} & \numprint{2.34} & \numprint{2.52} & \numprint{1.43} & \numprint{1.48} & \numprint{0.78} & \numprint{1.22} & \numprint{0.42} & \numprint{0.89}\\
loc-gowalla\_edges & \numprint{2.82} & \numprint{2.23} & \numprint{2.34} & \numprint{0.79} & \numprint{1.11} & \numprint{0.46} & \numprint{0.53} & \numprint{0.28} & \numprint{0.31} & \numprint{0.14} & \numprint{0.17} & \numprint{0.08} & \numprint{0.11}\\
web-NotreDame & \numprint{7.66} & \numprint{6.34} & \numprint{6.81} & \numprint{3.18} & \numprint{3.52} & \numprint{1.40} & \numprint{1.58} & \numprint{0.76} & \numprint{0.99} & \numprint{0.52} & \numprint{0.66} & \numprint{0.22} & \numprint{0.60}\\
web-Stanford & \numprint{34.62} & \numprint{28.81} & \numprint{29.51} & \numprint{14.27} & \numprint{14.01} & \numprint{5.73} & \numprint{8.68} & \numprint{4.41} & \numprint{5.50} & \numprint{1.80} & \numprint{2.13} & \numprint{1.18} & \numprint{1.50}\\
petster-dog-household & \numprint{5.31} & \numprint{4.24} & \numprint{3.68} & \numprint{2.21} & \numprint{1.95} & \numprint{1.05} & \numprint{1.14} & \numprint{0.66} & \numprint{0.75} & \numprint{0.62} & \numprint{0.78} & \numprint{0.46} & \numprint{0.73}\\
flixster & \numprint{13.99} & \numprint{10.30} & \numprint{11.02} & \numprint{6.03} & \numprint{6.00} & \numprint{3.01} & \numprint{3.31} & \numprint{2.25} & \numprint{1.97} & \numprint{1.36} & \numprint{1.45} & \numprint{1.33} & \numprint{1.90}\\
as-skitter & \numprint{17.14} & \numprint{13.26} & \numprint{14.23} & \numprint{7.59} & \numprint{8.13} & \numprint{4.11} & \numprint{4.42} & \numprint{2.42} & \numprint{2.44} & \numprint{1.50} & \numprint{2.08} & \numprint{1.13} & \numprint{1.74}\\
actor-collaboration & \numprint{8.69} & \numprint{6.28} & \numprint{5.48} & \numprint{3.28} & \numprint{3.27} & \numprint{1.85} & \numprint{1.71} & \numprint{1.06} & \numprint{1.03} & \numprint{0.68} & \numprint{0.69} & \numprint{0.48} & \numprint{1.01}\\
soc-pokec-relationships & \numprint{25.38} & \numprint{16.22} & \numprint{17.18} & \numprint{8.77} & \numprint{9.38} & \numprint{5.57} & \numprint{5.79} & \numprint{3.40} & \numprint{2.98} & \numprint{2.15} & \numprint{2.00} & \numprint{1.40} & \numprint{2.49}\\
soc-LiveJournal1 & \numprint{54.91} & \numprint{35.60} & \numprint{40.49} & \numprint{22.70} & \numprint{18.93} & \numprint{13.13} & \numprint{14.15} & \numprint{7.04} & \numprint{8.28} & \numprint{4.69} & \numprint{6.32} & \numprint{3.29} & \numprint{5.72}\\
livejournal-links & \numprint{62.27} & \numprint{44.49} & \numprint{44.85} & \numprint{25.44} & \numprint{24.65} & \numprint{12.45} & \numprint{15.29} & \numprint{7.69} & \numprint{8.53} & \numprint{5.07} & \numprint{6.89} & \numprint{3.77} & \numprint{6.95}\\
wikipedia\_link\_sh & \numprint{41.54} & \numprint{17.94} & \numprint{21.49} & \numprint{9.81} & \numprint{11.86} & \numprint{4.89} & \numprint{7.04} & \numprint{2.78} & \numprint{3.88} & \numprint{1.96} & \numprint{1.68} & \numprint{1.55} & \numprint{2.26}\\
wikipedia\_link\_sr & \numprint{56.30} & \numprint{45.93} & \numprint{41.54} & \numprint{24.91} & \numprint{23.70} & \numprint{11.36} & \numprint{13.18} & \numprint{7.09} & \numprint{6.54} & \numprint{4.44} & \numprint{5.11} & \numprint{2.95} & \numprint{5.23}\\
\end{tabular}

\end{subtable}
\end{table}

\begin{table}
\caption{Absolute running times (s) on expensive instances.
	Total: \adapt with preprocessing on a single core.}
\label{tab:expensive_rt} 
\bigskip

\begin{subtable}{\textwidth}
	\caption{OMP: OpenMP baseline,
		L: local-frame}\bigskip
	\label{tab:sub_rt1}
	
	\centering
	\scriptsize
	\begin{tabular}{lr | rr | rr | rr | rr | rr | rr}
Network Name &  & \multicolumn{2}{c|}{1 core} & \multicolumn{2}{c|}{2 cores} & \multicolumn{2}{c|}{4 cores} & \multicolumn{2}{c|}{8 cores} & \multicolumn{2}{c|}{16 cores} & \multicolumn{2}{c}{32 cores}\\
 & Total & OMP & L & OMP & L & OMP & L & OMP & L & OMP & L & OMP & L\\
\midrule\midrule
dimacs9-NY & \numprint{249}  & \numprint{246} & \numprint{212} & \numprint{97} & \numprint{106} & \numprint{54} & \numprint{55} & \numprint{30} & \numprint{28} & \numprint{19} & \numprint{14} & \numprint{10} & \numprint{6}\\
dimacs9-COL & \numprint{405}  & \numprint{397} & \numprint{358} & \numprint{177} & \numprint{177} & \numprint{101} & \numprint{94} & \numprint{55} & \numprint{47} & \numprint{30} & \numprint{23} & \numprint{17} & \numprint{11}\\
roadNet-PA & \numprint{1961}  & \numprint{1937} & \numprint{1851} & \numprint{1027} & \numprint{942} & \numprint{521} & \numprint{458} & \numprint{284} & \numprint{235} & \numprint{148} & \numprint{121} & \numprint{92} & \numprint{59}\\
roadNet-TX & \numprint{1965}  & \numprint{1937} & \numprint{2001} & \numprint{1042} & \numprint{1035} & \numprint{544} & \numprint{496} & \numprint{279} & \numprint{250} & \numprint{165} & \numprint{130} & \numprint{89} & \numprint{64}\\
dbpedia-all & \numprint{412}  & \numprint{402} & \numprint{395} & \numprint{227} & \numprint{215} & \numprint{126} & \numprint{80} & \numprint{76} & \numprint{42} & \numprint{54} & \numprint{22} & \numprint{41} & \numprint{19}\\
wikipedia\_link\_ceb & \numprint{1337}  & \numprint{1272} & \numprint{1435} & \numprint{701} & \numprint{707} & \numprint{415} & \numprint{307} & \numprint{238} & \numprint{160} & \numprint{156} & \numprint{98} & \numprint{121} & \numprint{74}\\
wikipedia\_link\_ru & \numprint{142}  & \numprint{126} & \numprint{132} & \numprint{89} & \numprint{73} & \numprint{52} & \numprint{44} & \numprint{36} & \numprint{23} & \numprint{26} & \numprint{12} & \numprint{24} & \numprint{12}\\
wikipedia\_link\_de & \numprint{155}  & \numprint{145} & \numprint{182} & \numprint{106} & \numprint{100} & \numprint{54} & \numprint{57} & \numprint{37} & \numprint{21} & \numprint{25} & \numprint{12} & \numprint{20} & \numprint{13}\\
wikipedia\_link\_it & \numprint{152}  & \numprint{112} & \numprint{145} & \numprint{82} & \numprint{70} & \numprint{58} & \numprint{41} & \numprint{27} & \numprint{24} & \numprint{20} & \numprint{14} & \numprint{16} & \numprint{9}\\
wikipedia\_link\_sv & \numprint{444}  & \numprint{423} & \numprint{472} & \numprint{258} & \numprint{255} & \numprint{160} & \numprint{105} & \numprint{105} & \numprint{61} & \numprint{73} & \numprint{32} & \numprint{60} & \numprint{31}\\
wikipedia\_link\_fr & \numprint{194}  & \numprint{168} & \numprint{177} & \numprint{131} & \numprint{106} & \numprint{75} & \numprint{58} & \numprint{41} & \numprint{30} & \numprint{32} & \numprint{14} & \numprint{24} & \numprint{12}\\
orkut-links & \numprint{206}  & \numprint{107} & \numprint{110} & \numprint{66} & \numprint{64} & \numprint{44} & \numprint{35} & \numprint{27} & \numprint{19} & \numprint{19} & \numprint{10} & \numprint{15} & \numprint{9}\\
\end{tabular}

\end{subtable}
\begin{subtable}{\textwidth}
	\label{tab:sub_rt2}
	\caption{S: shared-frame,
		I: indexed-frame}\bigskip

	\centering
	\scriptsize
	\begin{tabular}{lr | rr | rr | rr | rr | rr | rr}
Network Name &  & \multicolumn{2}{c|}{1 core} & \multicolumn{2}{c|}{2 cores} & \multicolumn{2}{c|}{4 cores} & \multicolumn{2}{c|}{8 cores} & \multicolumn{2}{c|}{16 cores} & \multicolumn{2}{c}{32 cores}\\
 & Total & S & I & S & I & S & I & S & I & S & I & S & I\\
\midrule\midrule
dimacs9-NY & \numprint{249}  & \numprint{198} & \numprint{203} & \numprint{95} & \numprint{101} & \numprint{49} & \numprint{53} & \numprint{26} & \numprint{31} & \numprint{13} & \numprint{14} & \numprint{6} & \numprint{9}\\
dimacs9-COL & \numprint{405}  & \numprint{345} & \numprint{339} & \numprint{177} & \numprint{171} & \numprint{94} & \numprint{93} & \numprint{50} & \numprint{49} & \numprint{22} & \numprint{25} & \numprint{11} & \numprint{17}\\
roadNet-PA & \numprint{1961}  & \numprint{1975} & \numprint{1781} & \numprint{998} & \numprint{950} & \numprint{471} & \numprint{463} & \numprint{242} & \numprint{252} & \numprint{113} & \numprint{151} & \numprint{60} & \numprint{82}\\
roadNet-TX & \numprint{1965}  & \numprint{1994} & \numprint{1943} & \numprint{1022} & \numprint{1038} & \numprint{504} & \numprint{514} & \numprint{232} & \numprint{277} & \numprint{119} & \numprint{165} & \numprint{63} & \numprint{89}\\
dbpedia-all & \numprint{412}  & \numprint{411} & \numprint{391} & \numprint{214} & \numprint{205} & \numprint{96} & \numprint{83} & \numprint{49} & \numprint{66} & \numprint{18} & \numprint{25} & \numprint{14} & \numprint{36}\\
wikipedia\_link\_ceb & \numprint{1337}  & \numprint{1421} & \numprint{1268} & \numprint{691} & \numprint{665} & \numprint{370} & \numprint{399} & \numprint{158} & \numprint{162} & \numprint{86} & \numprint{135} & \numprint{68} & \numprint{117}\\
wikipedia\_link\_ru & \numprint{142}  & \numprint{134} & \numprint{125} & \numprint{67} & \numprint{69} & \numprint{42} & \numprint{35} & \numprint{25} & \numprint{26} & \numprint{14} & \numprint{19} & \numprint{11} & \numprint{16}\\
wikipedia\_link\_de & \numprint{155}  & \numprint{155} & \numprint{170} & \numprint{80} & \numprint{98} & \numprint{39} & \numprint{56} & \numprint{24} & \numprint{32} & \numprint{12} & \numprint{19} & \numprint{8} & \numprint{15}\\
wikipedia\_link\_it & \numprint{152}  & \numprint{120} & \numprint{129} & \numprint{74} & \numprint{68} & \numprint{38} & \numprint{45} & \numprint{21} & \numprint{26} & \numprint{11} & \numprint{20} & \numprint{10} & \numprint{16}\\
wikipedia\_link\_sv & \numprint{444}  & \numprint{438} & \numprint{429} & \numprint{272} & \numprint{246} & \numprint{111} & \numprint{108} & \numprint{55} & \numprint{58} & \numprint{30} & \numprint{46} & \numprint{24} & \numprint{54}\\
wikipedia\_link\_fr & \numprint{194}  & \numprint{181} & \numprint{168} & \numprint{96} & \numprint{92} & \numprint{55} & \numprint{60} & \numprint{30} & \numprint{35} & \numprint{15} & \numprint{29} & \numprint{14} & \numprint{20}\\
orkut-links & \numprint{206}  & \numprint{119} & \numprint{107} & \numprint{60} & \numprint{66} & \numprint{34} & \numprint{37} & \numprint{19} & \numprint{25} & \numprint{10} & \numprint{20} & \numprint{7} & \numprint{15}\\
\end{tabular}

\end{subtable}
\end{table}

\end{document}